\documentclass[a4paper,UKenglish,cleveref, autoref, thm-restate]{lipics-v2021}
\nolinenumbers
\title{Algorithms for the local and the global postage stamp problem}


\newcommand{\LJKadd}{Univ.\ Grenoble Alpes, Laboratoire Jean Kuntzmann, UMR
  CNRS 5224, 150 place du Torrent, IMAG - CS 40700, 38058 Grenoble
  cedex 9, France}

\author{L\'eo Colisson Palais}{\LJKadd}{leo.colisson-palais@univ-grenoble-alpes.fr}{https://orcid.org/0000-0001-8963-4656}{}
\author{Jean-Guillaume Dumas}{\LJKadd}{Jean-Guillaume.Dumas@univ-grenoble-alpes.fr}{https://orcid.org/0000-0002-2591-172X}{}
\author{Alexis Galan}{\LJKadd}{Alexis.Galan@univ-grenoble-alpes.fr}{}{}
\author{Bruno Grenet}{\LJKadd}{BrunoGrenet@univ-grenoble-alpes.fr}{https://orcid.org/0000-0003-2057-5429}{}
\author{Aude Maignan}{\LJKadd}{Aude.Maignan@univ-grenoble-alpes.fr}{https://orcid.org/0000-0002-0905-2515}{}

\authorrunning{L. Colisson Palais, J-G. Dumas, A. Galan, B. Grenet,
  A. Maignan}
 \Copyright{L\'eo Colisson Palais, Jean-Guillaume Dumas, Alexis Galan,
   Bruno Grenet, Aude Maignan}

\hideLIPIcs  

\usepackage[colorinlistoftodos,backgroundcolor=yellow]{todonotes}

\usepackage{xspace}
\usepackage{booktabs}
\usepackage{multirow}

\usepackage{robust-externalize}\robExtConfigure{rename backup files for arxiv}
\robExtConfigure{
  enable fallback to manual mode
}

\usepackage[createShortEnv,commandRef=Cref]{proof-at-the-end}
\pratendSetGlobal{
  text link=,
  restate, 
}
\newcommand{\ProofInAppendix}[1]{{Proof in~\cref{#1}}}

\usepackage{algorithm}
\usepackage[noend]{algcompatible}
\newcommand{\To}{\textbf{to}\xspace}
\newcommand{\DownTo}{\textbf{down-to}\xspace}
\newcommand{\algorithmicreturn}{\textbf{return}}
\newcommand{\RETURN}{\STATE\algorithmicreturn{}\xspace}

\NewDocumentCommand{\bitwiseand}{}{\mathbin{\&}}
\RenewDocumentCommand{\COMMENT}{sm}{\IfBooleanTF{#1}{\hfill}{}{\footnotesize\textit{\(\triangleright\) #2}}}

\newcommand{\N}{{\mathbb{N}}}
\newcommand{\GStamps}{\href{https://github.com/jgdumas/GStamps}{GStamps}\xspace}

\usepackage{color,svgcolor}
\usepackage{hyperref}
\hypersetup{%
  breaklinks=true,
  colorlinks=true,
  linkcolor=darkgreen,
  citecolor=darkred,
  urlcolor=blue,
  hypertexnames=false,
  naturalnames=true,
}
\usepackage{cleveref}
\AddToHook{cmd/appendix/before}{%
    \crefalias{section}{appendix}%
    \crefalias{subsection}{appendix}
}

\begin{document}
\maketitle
\begin{abstract}
We consider stamps with different values (denominations) and same
dimensions, and an envelope with a fixed maximum number of stamp
positions.
The local postage stamp problem is to find the smallest value that
cannot be realized by the sum of the stamps on the envelope.
The global postage stamp problem is to find the set of denominations
that maximize that smallest value for a fixed number of distinct
denominations.
The local problem is NP-hard and we propose here a novel algorithm
that improves on both the time complexity bound and the amount of
required memory.
We also propose a polynomial approximation algorithm for the global
problem together with its complexity analysis.
Finally we show that our algorithms allow to improve secure
multi-party computations on sets via a more efficient homomorphic
evaluation of polynomials on ciphered values.
\end{abstract}
\section{Introduction}\label{sec:intro}
\pratendSetLocal{category=intro}
The postage stamp problem combines number theory and combinatorial
optimization.
It arises from a simple and practical question : given a limited set of
postage stamp denominations $A_k=\{a_1,a_2, \ldots , a_k\}$ and at
most $s$ places to stick the stamps, what is the first postage rate that cannot be formed?
We formally define the maximal postage rate, namely the $s$-range, in~\cref{def:srange}.
\begin{definition}\label{def:srange}
For a given positive integers $s$ and
 a set of $k$ positive and all different integers
$A_k = \{a_1,a_2, \cdots, a_k \}$, with $a_1<a_2<\ldots<a_k$,
the \emph{$s$-range} is the largest integer $n$ such that:
\(\forall{i}\leq{n},\exists(\lambda_1,\ldots,\lambda_k)\in\N^k,\sum_{j=1}^k\lambda_j\leq{s},i=\sum_{j=1}^k\lambda_ja_j\);
$n$ is then denoted $n_s(A_k)$.
Then the \emph{extremal $s$-range} is defined as
\(n_s(k)=\max_{A_k}(n_s(A_k))\) and a basis $A_k$ satisfying
$n_s(A_k)=n_s(k)$ is called an \emph{extremal basis}.
\end{definition}



We focus on the two main problems related to postage stamps, namely
\begin{definition}
The {\em local} postage stamp problem (LPSP) is the determination of the
$s$-range for~$s$ stamps of a given basis of size~$k$.
\end{definition}
\begin{definition}
The {\em global} postage stamp problem (GPSP) is the determination of an
extremal basis for given parameters~$k$ and~$s$.
\end{definition}

%
For example, if we choose $s=2$ and $k=3$, we have $n_2(3)=n_2(\{1,3,4\})=8$.
Moreover $n_1(k)=k$ and $n_s(1)=s$.
%
%
There is a simple upper bound on $n$: 
\begin{lemma}\label{lem:cbound} For $A_k=\{a_1,\cdots,a_k\}$ a stamp basis, we have, for all $s$, $n_s(A_k) \leq{sa_k}$.
\end{lemma}





Historically, Hans Rohrbach \cite{rohrbach1937beitrag} was the first to mention the postage stamp problem in the case where $s=2$.Then, a  lot of effort was devoted to the search for extremal
bases~\cite{rohrbach1937beitrag,Lunnon:1969:Postage,mrose_untere_1979,Selmer:1980:Postage,Selmer:1983:Postage,Selmer:1985:Postage,mitchell_another_1989,kirfel_extremal_1989,kirfel_extremal_1990,selmer_associate_1992,Challis:1993:Two,mossige_postage_2001,yu_new_2015}.
For instance,
\cite{challis_extremal_2010}, provides a survey of known
extremal $s$-range and extremal bases for small values of~$k$ and~$s$.
Then~\cite{shallit_computational_2002} shows that LPSP is NP-hard
under Turing reductions, but can be solved in polynomial time if k is
fixed.
Nonetheless, it is possible to give polynomial-time approximation
algorithms that compute interesting basis ; even if evaluating exactly
their $s$-range can be hard.

{\bf Contributions}.
The goal of this paper is to provide efficient
\emph{algorithms} for the both the local and the global problems.
We first survey the different known methods that can compute "good" basis
for the global problem : this means that for given a number of
denominations $k$ and positions $s$, we obtain a basis of
denominations with an $s$-range with the same asymptotic than the
(unknown) extremal basis.
We show that a recursive divide \& conquer algorithm, in which we
combine smaller basis, is the best \emph{polynomial}
strategy, among these, to generate such good bases.
We then prove that, asymptotically, the best recursion is where the
smaller bases are balanced.
For some basis, however, if one can afford more computational effort,
there are examples where a dynamic programming, for instance, can find
better bases, by choosing the best cut.

Then, we also propose a novel algorithm, asymptotically and
practically better than existing methods for the local problem.
In particular, with respect, e.g., to Mossige's
s-range~\cite{Mossige:1981:hrange}, we were able, first, to remove
an asymptotic dependency in $s$, and, second, to also remove an $s$
factor in the required memory.
We have implemented all the different algorithms in the
\GStamps library\footnote{\url{https://github.com/jgdumas/GStamps}}.

Finally, we propose a way to improve in practice some cryptographic
protocols relying on \emph{private} polynomial evaluation, where the
computations are made using fully homomomorphic encryption.

{\bf Paper organization}.
\Cref{sec:polytimeapprox} presents a compared analysis of the
Fibonacci sequence, the Alter and Barnett sequence, and the recursive
divide \& conquer algorithm for the approximation of the global stamp
problem.
Then, in~\cref{sec:algoforlargeenvelopes}, we give two novel
algorithms for the local problem that improve on the time
complexity bound and on the amount of required memory.
Finally, in~\cref{sec:applications} we show how good approximation
bases can speed up secure multi-party computations on sets via a more
efficient homomorphic evaluation of polynomials on ciphered values.
\section{Polynomial-time approximations for the global postage stamp}\label{sec:polytimeapprox}
\pratendSetLocal{category=polytimeapprox}

This section is devoted to approximation algorithms for the global problem :
Given $k$ and $s$, compute a lower bound on $n_s(k)$ that is as close as
possible to the correct (unknown) value. This goes through producing some basis
$A_{k,s}$ with $k$ denominations and proving a lower bound on $n_s(A_{k,s})$.

We review some existing constructions. First we focus on a construction of Alter
and Barnett~\cite{alter_remarks_1977}. They provide, for any $k \ge s$, a pretty
good basis $A_{k,s}$. In particular in the case $k = s$, $A_{k,s}$ is made of
the first $k$ even-indexed Fibonacci numbers. They provide the exact value for
$n_s(A_{k,s})$. We extend their construction to the case $k \le s$, and slightly
improve the construction in the case $k \ge s$. Although they provide the
correct value for $n_s(A_{k,s})$, their proof only prove the lower bound. We
provide a full proof of the upper bound that applies both to their original
construction and our improved version. We also analyse the asymptotics of
$n_s(A_{k,s})$ in both cases.

We then develop and analyse an idea of Mrose~\cite{mrose_rekursives_1974} that
builds a good basis $A_{k,s}$ with a recursive algorithm. We also analyse the
effect of using a database of small good base cases. Combining several
constructions, the database of base cases, we obtain the best known asymptotic
constructions as well as experimental validation of these results.
Using our \GStamps library, one can retrieve all the result presented in this part. Namely, our algorithms \textit{fibo, alba, greedy, geom} compute the $s$-ranges for (respectively) the Fibonacci, the Alter and Barnett and its variant with respect to~\cref{prop:balgreedy}, and the geometric basis constructions, for chosen $k,s$. Our algorithm \textit{basis} computes the best found basis for given $k,s$ using the recursive algorithm.

\subsection{Fibonacci sequence}

We first describe the Alter and Barnett's construction focus on the case~$k=s$. Let $(f_i)_{i\ge 0}$ be the Fibonacci sequence defined by
$f_0 = 0$, $f_1 = 1$ and $f_{i+2} = f_i+f_{i+1}$ for $i \ge 0$.
The Fibonacci stamps is the set $F_k=\{f_{2i}\}_{i\in[1,k]}$. Alter and
Barnett~\cite{alter_remarks_1977} state that $n_k(F_k) = f_{2k+1}-1$, while they
only prove that $n_k(F_k) \ge f_{2k+1}-1$.

\begin{propositionE}[\ProofInAppendix{sec:proofsapprox}]\label{thm:fibo}
The $k$-range of~$F_k$ is~$n_k(F_k)=f_{2k+1}-1$.
\end{propositionE}
\begin{proofE}
    We first prove by induction on $k$ that $n_k(F_k) \ge f_{2k+1}-1$. For $k = 1$, $n_1(F_1) = n_1(\{1\}) = 1 = f_3-1$.

    Consider now $k > 1$. By induction hypothesis, $n_k(F_{k-1}) \ge f_{2k-1}-1$. We prove that every integer $m < f_{2k+1}$ can be written as $\sum_{i=1}^k\lambda_i f_{2i}$ with $\sum_i\lambda_i \le k$. If $m < f_{2k-1}$, this is true by induction hypothesis. If $f_{2k-1} \le m < f_{2k}$, $m-f_{2k-2} < f_{2k}-f_{2k-2} = f_{2k-1}$. Therefore, $m-f_{2k-2}$ can be written with $(k-1)$ stamps by induction hypothesis. Hence $m$ can be written with $k$ stamps with the additional stamp $f_{2k-2}$. Similarly if $f_{2k} \le m < f_{2k+1}$, $m-f_{2k} < f_{2k+1}-f_{2k} = f_{2k-1}$ and $m$ can be obtained using the stamp $f_{2k}$ and $(k-1)$ other stamps.

    We now prove that $n_k(F_k) < f_{2k+1}$, that is $f_{2k+1}$ cannot be written as $\sum_{i=1}^k \lambda_i f_{2i}$ with $\sum_i\lambda_i \le k$. We assume otherwise and rewrite the sum as $f_{2k+1} = \sum_{i=1}^{2k}\lambda_i f_i$ where odd indexed $\lambda_i$s vanish. We describe a rewriting rule that keeps $\sum_i\lambda_i$ non-increasing, and let the following property invariant:
    \begin{itemize}
        \item[$(\mathcal P)$] let $j = \max\{i:\lambda_i > 1\}$, then $\lambda_{j+1} = 0$, $\lambda_i\in\{0,1\}$ and $\lambda_i\lambda_{i+1} = 0$ for $i > j+1$, and every odd indexed $\lambda_i$ for $i < j$ vanishes. (If no such $j$ exists, $j = 0$ by convention.)
    \end{itemize}
    The rewriting rule is as follows: Consider $\lambda_j>1$ as in $(\mathcal P)$. Using the identity $2f_j = f_{j-2} + f_{j+1}$, we replace $\lambda_j$ by $\lambda_j-2$, increase $\lambda_{j-2}$ and set $\lambda_{j+1} = 1$. If $\lambda_{j+2} = 1$, we use the identity $f_{j+1}+f_{j+2} = f_{j+3}$ to set $\lambda_{j+1}=\lambda_{j+2} = 0$ and $\lambda_{j+3} = 1$, and so on until no two consecutive $\lambda_i$s, $i > j+1$, are non-zero. If $\lambda_{j+2} = 0$ but $\lambda_j$ is still non-zero, we apply the same process from the pair $(\lambda_j,\lambda_{j+1})$. It is readily checked that property $(\mathcal P)$ holds after one step of rewriting.

    We apply the rewriting rule until one of the two events occur: Either no $j$ such that $\lambda_j > 1$ exists, or $j = 2$. In the first case, $f_{2k+1}$ is a sum of non-consecutive non-repeated Fibonacci numbers. This is impossible by Zeckendorf theorem, or simply by noting that the largest such sum is $\sum_{i=1}^k f_{2i} = f_{2k+1}-1$. In the second case, $f_{2k+1} = \lambda_2f_2 + \sum_{i\ge 4} \lambda_i f_i$ with $\sum_i \lambda_i \le k$ and $\lambda_2 > 1$. This implies that $\lambda_i = 0$ for $i > 2k$. Since there are at most $k-\lambda_2$ non-zero terms in the sum $\sum_{i\ge 4} \lambda_i f_i$, its value is at most $\sum_{\ell=\lambda_2+1}^k f_{2\ell} = f_{2k+1}-f_{2(\lambda_2+1)+1}$. But $\lambda_2 < f_{2(\lambda_2+1)+1)}$ and the sum cannot reach $f_{2k+1}$. In both cases, we obtain a contradiction.
\end{proofE}
Let $\varphi=\varphi_{\!_+}=\frac{1+\sqrt{5}}{2}\approx{1.618}$ be the
golden number, then $\varphi^2=1+\varphi$ and let
$\varphi_{\!_-}=\frac{1-\sqrt{5}}{2}=1-\varphi=-\varphi^{-1}$.
Then we have that:
\begin{equation}\label{eq:fibo}
f_n=\frac{1}{\sqrt{5}}(\varphi_{\!_+}^n-\varphi_{\!_-}^n)=
\left\lfloor\frac{1}{\sqrt{5}}\varphi^n\right\rceil
\end{equation}

Now this basis can be used also if $k\leq{s}$ by just adding integers
obtained by groups of $k$.
\Cref{cost:fibo} shows that the
resulting $s$-range then grows accordingly.

\begin{corollaryE}[\ProofInAppendix{sec:proofsapprox}]\label{cost:fibo}
For $k\leq{s}$, let $s=qk+r$ by Euclidean division, then:
 \(\left\lceil\frac{s}{\sqrt{5}}\right\rceil(1+\varphi)^k >
 n_s(F_k)>
\left(\left\lfloor\frac{s}{k}\right\rfloor\left(1-\frac{\sqrt{5}}{\varphi^{2k+1}}\right)+\frac{1}{\varphi^{2k-2r}}\right)\frac{\varphi}{\sqrt{5}}(1+\varphi)^k.\)
\end{corollaryE}
\begin{proofE}
For the upper bound, by~\cref{lem:cbound} with basis
$F_k$, we get $n_s(F_k)<sf_{2k}$.
Then, by~\cref{eq:fibo},
$f_{2k}=\left[\frac{1}{\sqrt{5}}\varphi^{2k}\right]=\left[\frac{1}{\sqrt{5}}(1+\varphi)^{k}\right]$.
For the lower bound, consider that
each group of $k$ stamps can reach $f_{2k+1}-1$ by~\cref{cost:fibo}.
Then the last group with $r$ stamps can reach at least $f_{2r+1}-1$.
Therefore, at least all integers up to
$q(f_{2k+1}-1)+f_{2r+1}-1$ can be reached.
This is
\(\left(q-\frac{1}{f_{2k+1}}+\frac{f_{2r+1}}{f_{2k+1}}\right)f_{2k+1}-1\),
from which one gets the announced lower bound using~\cref{eq:fibo},
with $q=\lfloor\frac{s}{k}\rfloor$.
\end{proofE}

\subsection{Greedy piece-wise linear algorithm}
Alter and Barnett provide a general construction for $k \ge s$. Writing $k =
qs+r$ with $r < s$, they defined a basis $G_k$ as the concatenation of $s$
blocks $B_1$, \dots, $B_s$ such that $|B_i| = q$ for $i < s$ and $|B_s| = q+r$,
and the denomination within each block are in arithmetic progression.

For each $i\in[s]$, let $G'_i=B_1\Vert\cdots\Vert B_i$ and let $U_i =
n_i(G'(i))$. They state the following induction relation:
\begin{equation}\label{eq:ABgreedy}
\text{\cite[Eq. (20)]{alter_remarks_1977}} \quad
U_{t+2}=(q+2)U_{t+1}-U_t+q,U_0=0,U_1=q.
\end{equation}
They deduce that
\begin{equation}\label{eq:ABbound}
n_s(G_k) = \sum_{i=1}^s\binom{s+i}{2i}q^i+r\sum_{i=0}^{s-1}\binom{s+i-1}{2i}q^i.
\end{equation}

As mentioned already, although this value is correct, they only prove
that $n_s(G_k)$ is at least the right-hand side of~\eqref{eq:ABbound}. We
provide in \cref{sec:proofsapprox} a full proof of the upper bound.

Solving exactly~\cref{eq:ABgreedy}, one gets that the dominant term of
\cref{eq:ABbound} is in fact:
\(\left[\frac{1}{2}\left(1-\sqrt{1-\frac{4}{q+4}}\right)+\frac{q}{\sqrt{q}\sqrt{q+4}}\right]\left(1+q\frac{1+\sqrt{1+\frac{4}{{q}}}}{2}\right)^s\).
The common ratio boils down to $1+\varphi$ when $q=1$.
Otherwise, the ratio is asymptotically close to $1+q=1+\frac{k}{s}$. 
With $\epsilon_q=\frac{\sqrt{1+4/q}-1}{2}$, we thus obtain a basis
whose $s$-range is asymptotically 
\begin{equation}\label{cost:greedy}
\alpha\left(1+\frac{k}{s}(1+\epsilon_q)\right)^s.
\end{equation}

Actually, we can improve Alter and Barnett's construction. The idea is to better
balance the blocks in $G_k$. Instead of using $s-1$ blocks of size $q$ and one of size $q+r$, one can use $(s-r)$ blocks of size $q$ and $r$
of size $q+1$, with the same construction of the blocks. The construction
provides a larger asymptotic for $n_s(G_k)$.

\begin{propositionE}[\ProofInAppendix{sec:proofsapprox}]\label{prop:balgreedy}
    Using the variant of Alter and Barnett's construction with $r$
    blocks of size $(q+1)$ and then $(s-r)$ blocks of
    size $q$, the dominant term of $n_s(G_k)$
    becomes
    \(\left(1+(q+1)(1+\epsilon_{q+1})\right)^r\left(1+q(1+\epsilon_q)\right)^{s-r}.\)
\end{propositionE}
\begin{proofE}
Use~\cref{eq:ABgreedy} with $q\leftarrow{q+1}$ and $s\leftarrow{r}$, for
the first $r$ blocks;
then use~\cref{eq:ABgreedy} again, now with $q$ and
$(U_0,U_1)=(U_{r-1},U_r)$ from the previous blocks.
\end{proofE}

\subsection{Recursive algorithm}
From two sub-bases $A_{k_1},B_{k_2}$ with respective $s$-ranges
$n_{s_1}(A_{k_1})$ and $n_{s_2}(B_{k_2})$, using
\cite[Hilfssatz~2]{mrose_rekursives_1974}, we can build a basis $C_k$
of size $k=k_1+k_2$ such that, for $s=s_1+s_2$:
\(n_{s}(C_k)\geq(n_{s_1}(A_{k_1})+1)(n_{s_2}(B_{k_2})+1)-1\).
As this relation is true for any basis, we more precisely have the
following~\cref{prop:Mrose}.
\begin{propositionE}[\ProofInAppendix{sec:proofsapprox}]\label{prop:Mrose}
\(n_{s_1+s_2}(k_1+k_2)\geq (n_{s_1}(k_1)+1)(n_{s_2}(k_2)+1)-1\)
\end{propositionE}
\begin{proofE} We recall the proof of~\cite{mrose_rekursives_1974}.
Let $A_{k_1}=\{a_1,\cdots,a_{k_1}\}$,
$B_{k_2}=\{b_1,\cdots,b_{k_2}\}$, $n_1=n_{s_1}(A_{k_1})$ and
$n_2=n_{s_2}(B_{k_2})$. We define $C_k$ as the concatenation of
$A_{k_1}$ and all the denominations of $B_{k_2}$ multiplied by
$n_1+1$, i.e. $C_k=A_{k_1}\Vert(n_1+1)B_{k_2}$. It is easy to see that
we can generate $i(n_1+1)$ for $i\in[n_2]$ with at most $s_2$ stamps,
using the stamps from $(n_1+1)B_{k_2}$. We also can generate all
$j\in[n_1]$ with at most $s_1$ stamps from $A_{k_1}$. Combined
together, we can obtain all $i(n_1+1)+j$ for $i\in[n_2]$ and
$j\in[n_1]$ with at most $s$ stamps.
Then, $n_s(C_k)\geq{n_2(n_1+1)+n_1}$.
By taking extremal bases for $A_{k_1}$ and $B_{k_2}$,
we obtain that $n_s(C_k)\geq{(n_{s_1}(k_1)+1)(n_{s_2}(k_2)+1)-1}$.
The result follows as by definition $n_{s}(k)\geq{n_s(C_k)}$.
\end{proofE}
The constructive proof of~\cref{prop:Mrose} thus provides a recursive Divide \& Conquer
algorithm computing a basis by recursively splitting both
the number of denominations and the number of stamps\footnote{This can
  give rise to a natural dynamic   programming version in order to
  determine the best cuts of $k$ and $s$ into $k=k_1+k_2$ and
  $s=s_1+s_2$, but then the algorithm is not  polynomial-time
  anymore.}.
In the rest of this section, we analyze the behavior of this recursive
Divide \& Conquer polynomial-time algorithm.
We start by determining the asymptotic, with the help of a simpler
function in~\cref{lem:po2}. This will provide a lower bound when
$k_1=k_2$ and~$s_1=s_2$ are both powers of two, dividing both
parameters by $2$, until one of~$k$ or~$s$ is $1$.

\begin{lemmaE}[\ProofInAppendix{sec:proofsapprox}]\label{lem:po2}
Let~$f:\N\rightarrow\N$ such that $f(0) = T$ and, for $i\ge 0$,
$f(i+1) = (f(i)+1)^2-1$. Then $f(i) = (1+T)^{2^i}-1$.
\end{lemmaE}
\begin{proofE}
\begin{enumerate}
\item First, we have that $T=1+T-1=(1+T)^{2^0}-1$;
\item Then, by induction, we directly obtain that
  $((1+T)^{2^i}-1+1)^2-1=(1+T)^{2^{i+1}}-1$.
\end{enumerate}
\end{proofE}

Now, suppose that $2^{i+j}=k\geq s=2^i$. We use a \emph{midpoint
  construction}, cutting both parameters in halves.
The recusive cuts then stops with the minimal case such that
$n_1(\frac{k}{s})=\frac{k}{s}$.
Thus, letting $T=\frac{k}{s}$ in~\cref{lem:po2}, we obtain the
following lower bound for a recursive Divide \& Conquer algorithm cutting in halves:
 \begin{equation}\label{cost:midpoint}
  n_s(k)\geq\left(1+\frac{k}{s}\right)^s-1.
\end{equation}
This is similar to the asymptotics of~\cref{cost:greedy}, but slightly
less interesting since there the common ratio is
$\left(1+q(1+\epsilon_q)\right)$, indeed slightly better than
$1+q=\left(1+\frac{k}{s}\right)$.

For the opposite case, namely $2^{i+j}=s\geq k=2^i$, now the minimal
case of the midpoint construction gives $n_\frac{s}{k}(1)=\frac{s}{k}$
resulting to the dual lower bound:
\begin{equation}\label{cost:rec}
  n_s(k)\geq\left(1+\frac{s}{k}\right)^k-1.
\end{equation}




This is now to be compared with~\cref{cost:fibo} with the same
exponent but a common ratio $(1+\varphi)$.
We see that for small values of $k$ the recursive algorithm will be
better, and then Fibonacci's algorithm will be better when $k$ is
closer to $s$ with ratio $\frac{s}{k}<\varphi$.

Now the lower bound of~\cref{cost:rec} can be refined by stopping the
recursion earlier, for instance at $k=2$.
Indeed, $n_s(2)=(s^2+6s+1)/4=\frac{s}{2}(\frac{s}{2}+3)+\frac{1}{4}$
(see for instance~\cite{alter_remarks_1977} and references therein).
We then obtain as minimal case
$n_\frac{2s}{k}(2)={\frac{s}{k}(\frac{s}{k}+3)+\frac{1}{4}}$.
Combined with~\cref{lem:po2}, this results in the following improved
bound:
\begin{equation}\label{cost:rectwo}
  n_s(k)\geq\left(1+\frac{s}{k}\left(\frac{s}{k}+3\right)+\frac{1}{4}\right)^{\frac{k}{2}}-1.
\end{equation}

\begin{remark}\label{rem:geom}
Note that a possible basis could be in geometric
progression $\{1,r,\ldots,r^{k-1}\}$, as
in~\cite[Eq.~(26)]{Stohr:1955:ZahlenreiheI}.
In this case, the best common ratio is $r\approx{1+\frac{s+1}{k}}$,
which yields a range never exceeding
$\left(1+\frac{s+1}{k}\right)^k-2$~\cite{Stohr:1955:ZahlenreiheI,Stohr:1955::ZahlenreiheII,Tripathi:2008:Postage}.
If $k=2$, this is not better than $(s^2+6s+1)/4$, and if $k\geq{3}$,
this is also always less interesting than the recursive algorithm
with~\cref{cost:rectwo}.
\end{remark}

As with~\cref{cost:rec,cost:rectwo}, the lower bound of~\cref{cost:midpoint} can be refined by
stopping the recursion earlier, for instance at $s=2$.
\cite{mrose_untere_1979} has indeed shown that
$n_2(k)\geq\frac{2}{7}k^2$.
Thus, $n_2(\frac{2k}{s})\geq\frac{8}{7}\left(\frac{k}{s}\right)^2$ and
we obtain:
\begin{equation}\label{cost:midpointtwo}
  n_s(k)\geq\left(1+\frac{8}{7}\left(\frac{k}{s}\right)^2\right)^{\frac{s}{2}}-1.
\end{equation}



We now turn to the general case, with $k$ and $s$ not powers of two.
Then, the following~\cref{lem:midpoint}, shows that asymptotically,
the best lower bound is given by any cut with the same fraction.
In particular, the midpoint construction of~\cref{lem:po2} thus
also provides the best lower bound when $k\leq{s}$.

\begin{lemmaE}[\ProofInAppendix{sec:proofsapprox}]\label{lem:midpoint}
Let $G(k,s)=(1+s/k)^k-1$ and
$H_{k,s}(t,u)=(G(t,u)+1)(G(k-t,s-u)+1)-1$ with $t\in\{1..k-1\}$ and
$u\in\{1..s-1\}$.
Then $G(k,s)$ is a local maximum of $H_{k,s}$.
\end{lemmaE}
\begin{proofE}
The first derivative of $g$ in $u$ is
$g_1=-t^{-t}(t+u)^{t-1}(k-t)^{t-k}(k-t+s-u)^{k-t-1}(ku-st)$ whose
zeroes are $k-t+s$ and $st/k$. With $t\in\{1..k-1\}$ and
$u\in\{1..s-1\}$,
the former is too large. Now the second derivative of $g$ at
$u=st/k$ is
$g_2=-\frac{1}{t}k^{3-k}(k+s)^{t-2}(k-t)^{t-k}((k-t)(k+s))^{k-t-2}(k+s)^2(k-t)$
and is thus negative when $t\in\{1..k-1\}$. Therefore, $u^*=st/k$ is local
maximum of $H_{k,s}$.

Finally, let $t=k/v$. With this, we get that $u^*=st/k=s/v$.
Then,
$H_{k,s}(k/v,s/v)=(G(\frac{k}{v},\frac{s}{v})+1)(G(k\frac{v-1}{v},s\frac{v-1}{v})+1)-1=(1+s/k)^{k/v}(1+s\frac{v-1}{v}(k\frac{v-1}{v})^{-1})^{k\frac{v-1}{v}}-1=(1+s/k)^{k/v+k-k/v}-1=G(k,s)$,
independently of the factor $v$, and the lemma is proven.
\end{proofE}

This analysis is of course quite rough for small values of $k$
and $s$, where the involved constants also matter.
For very small values, extremal bases and the extremal $s$-range are known,
see e.g.~\cite{challis_extremal_2010}, and can thus be fed to
the recursive algorithm before reaching $k=1$ or $s=1$, as was shown
for instance with~\cref{cost:rectwo,cost:midpointtwo}.

\begin{remark}\label{rq:DCfibo}
In practice this makes the recursive algorithm asymptotically better than
Fibonacci, even with $k=s$.
For instance indeed, the best known sequence for $k=s=5$ is
$\{1,4,9,31,51\}$ with $s$-range $126$ (while Fibonacci
gives $\{1,3,8,21,55\}$ with $s$-range $88$).
On the one hand, the recursive algorithm with $k=s=5\cdot2^i$,
stopping the recursion at threshold~$5$, has its lower bound
reaches at least
$(1+126)^{2^i}\approx{2.635}^{5\cdot2^i}$,
from~\cref{lem:po2}.
On the other hand, from~\cref{cost:fibo}, Fibonacci sequence will
never give a range greater than
$(1+\varphi)^{5\cdot2^i}\approx{2.618}^{5\cdot2^i}$.
This is also shown in practice in~\cref{fig:recfibo}.
\textEnd{
In~\cref{fig:recfibo}, we compare the recursive algorithm and the Fibonacci
  sequence, using the programs \texttt{./bin/basis} and
  \texttt{./bin/fibo} of the \GStamps library.
\begin{figure}[ht]\centering
\input{data_recfibo}
\caption{Range ratio between the recursive algorithm and the Fibonacci
  sequence, for~$k=s$}\label{fig:recfibo}
\end{figure}
}
\end{remark}

\begin{remark}\label{rq:DCgreedy}
The recursive algorithm with good base cases performs also better than
the Greedy algorithm.
We have seen that~\cref{cost:greedy} again gives a common ratio of
$1+\varphi$ when $k=s$ and not much better than $q=k/s$ otherwise,
namely $q(1+\epsilon_q)$.
The dominant term of the lower bound is thus
$(1+q(1+\epsilon_q))^s=q^s(1+\epsilon_q+\frac{1}{q})^s$,
with $1+\epsilon_q+\frac{1}{q}=1+\frac{1}{2}(\sqrt{1+4/q}-1)+\frac{1}{q}$
asymptotically close to
$1+\frac{1}{2}(1+2/q-1)+\frac{1}{q}=1+\frac{2}{q}$.
This is to be compared to~\cref{cost:midpointtwo},
$(1+\frac{8}{7}q^2)^{s/2}=q^s\sqrt{8/7+1/q^2}^s$,
where $\sqrt{8/7+1/q^2}=\sqrt{8/7}\sqrt{1+7/(8q^2)}$ is asymptotically close to
$\sqrt{8/7}+\sqrt{7/8}\frac{1}{2q^2}>1$.
\end{remark}

Overall, using~\cref{rq:DCfibo,rq:DCgreedy}, we
obtain~\cref{tab:regimes} that shows the dominant terms of cost of the
algorithms, depending on how $k$ and $s$ compare.

\begin{table}[htbp]\centering
\caption{Asymptotic regimes of the approximation algorithms for
  different $k$ and $s$}\label{tab:regimes}
\begin{tabular}{cc}
\toprule
$k\leq{s}$ & $k\geq{s}$ \\
\midrule
Fibonacci: $\frac{s}{k}(1+\varphi)^k $ & Greedy:
\(\left(1+\frac{k}{s}\right)^s\)\\
D\&C:
$\left(1+\frac{s}{k}\left(\frac{s}{k}+3\right)+\frac{1}{4}\right)^{\frac{k}{2}}$
&
D\&C:
$(1+\frac{8}{7}\left(\frac{k}{s}\right)^2)^{s/2}$
\\
\bottomrule
\end{tabular}
\end{table}

From~\cref{rq:DCgreedy}, the recursive algorithm with threshold is
thus asymptotically better.
It is shown in~\cref{fig:recbalgreed} that it is also better for not
so large values of $q=k/s$.
Indeed, in~\cref{fig:recbalgreed}, this is shown by the curve
"Mid-cut"
(except for the point $k=26,s=8$, since
even if the extremal bases for $k=12,s=4$ have $s$-range
$700$~\cite[Addendum]{challis_extremal_2010}, but our initial implementation
only finds a $s$-range of $566$ at
$k=13,s=4$).
In this figure, "Mid-cut" is obtained with the recursive algorithm with
threshold, always dividing at the middle point $k/2,s/2$.
In the second curve, "First-cut", we select the best cutting
point, for the first recursive call, with dynamic programming,
among $1..k-1,1..s-1$. It then recursively further divide always at the
middle point. Then "Best-cuts" uses a dynamic programming selection at
each recursive call.
Finally, the curve "Good-basis", incorporates, as base
cases to the midpoint construction of "Mid-cut", all the good basis
obtained in a programming
contests\footnote{Al Zimmermann's "Son Of Darts" \url{http://azspcs.com/Contest/SonOfDarts/FinalReport}}.
There, a basis at $k=13,4$ has $s$-ranges up to $878$,
improving the result for $k=26,s=8$.
Now, again, there is a drop at $k=62,s=8$, even for this curve
"Good-basis": this is now because the brute force searches of the
contest did provide good bases only until $k=30,s=4$. Therefore, after
$k=31$ we have a supplementary recursive call which start to reduce
the advantage obtained with better base cases.
Note that "Mid-cut'' and "Good-basis" represent polynomial-time
algorithms, implemented in \texttt{./bin/basis}, while "First-cut''
and "Best-cuts'', both implemented in \texttt{./bin/dynprg} with
different parameters, are not.

\begin{figure}[htbp]\centering
\input{data_recbg}
\caption{$s$-range ratio, Recursive with threshold
  vs~\cref{prop:balgreedy},~$s=8$}\label{fig:recbalgreed}
\end{figure}

\section{$s$-range algorithms for large envelopes}\label{sec:algoforlargeenvelopes}
\pratendSetLocal{category=algoforlargeenvelopes}

\subsection{Mossige s-range}
For the sake of completeness we give in \cref{appendix:proofsalgoforlargeenvelopes}, \cref{alg:mossige}, a version
of the algorithm of~\cite{Mossige:1981:hrange} computing the $s$-range
of a basis.
Roughly, the idea is to compute the values reached with each number of
possible stamps from $1$ to $s$ : with each additional stamp, all the
previous values are examined (at most $sa_k$), and the previously
reached added to one of $a_1$ to $a_k$ are then attainable. We thus
have~\cref{prop:mossigecost}.

\textEnd{ {For the sake of completeness we give in~\cref{alg:mossige}, a
  version of the algorithm of~\cite{Mossige:1981:hrange}.}
  \begin{algorithm}[H]
    \caption{ Mossige
s-range~\cite{Mossige:1981:hrange}}\label{alg:mossige}
    \begin{algorithmic}[1]
      \REQUIRE Basis $A_k=\{a_1=1<a_2<\ldots<a_k\}$, stamps $s$.
      \ENSURE The $s$-range $n_s(A_k)$.
      \STATE Let Reached~$=[0,\ldots,0]\in\{0,1\}^{a_ks}$;
      \FOR{$j=1$ \To $k$}
      \STATE Reached$[a_j]\leftarrow{1}$;
      \ENDFOR
      \FOR{$d=1$ \To $s-1$}
      \FOR{$i=a_kd$ \DownTo $d$}\COMMENT*{$a_kd$ is the largest
        reached for now}
      \IF{Reached$[i]==1$}
      \FOR{$j=1$ \To $k$}
      \STATE Reached$[i+a_j]\leftarrow{1}$;\COMMENT*{now reached with
        up to one more stamp}
      \ENDFOR
      \ENDIF
      \ENDFOR
      \ENDFOR
      \FOR{$i=2$ \To $a_ks$}
      \IF{Reached$[i]== 0$}
      \RETURN $i-1$;\COMMENT*{smallest value not reached}
      \ENDIF
      \ENDFOR
      \RETURN $a_ks$;
    \end{algorithmic}
  \end{algorithm}
}

\begin{proposition}[\cite{Mossige:1981:hrange}]\label{prop:mossigecost}
\Cref{alg:mossige} is correct and requires $O(ks^2a_k)$ assignments.
\end{proposition}

More precisely, the number of assignments (modifications of the
binary table storing whether each value is so far attainable or not)
is lower than $a_k\frac{s(s-1)}{2}k$, with $n_s(A_k) \leq sa_k$.

\subsection{Incremental LPSP algorithm}\label{subsec:incrementalLPSP}

We describe in \cref{alg:leo} a new algorithm computing the maximum
coverage given a fixed set of stamp values $a_1 < \dots < a_k$ and a
maximum number of allowed stamps $s$. This algorithm outperforms
\cite{Mossige:1981:hrange} both asymptotically and in practice when
$s$ is large enough, and is also arguably simpler. Our new algorithm
runs in time $O(k s a_k)$ instead of $O(k s^2 a_k)$, and we propose in
the next section an improvement to have a better asymptotic time
complexity $O(kn)$ with a smaller memory footprint scaling with
$O(a_k)$ instead of $O(n)$ that is needed in this first version and in
the s-range algorithm~\cite{Mossige:1981:hrange}.

Our first algorithm keeps two tables $T$ and $U$ where all the cells of
$T$ are initialized to a large enough number, say $s+1$, except for
$T[0] = 0$. The semantics of this array being that we know how to
generate the integer $i$ using $T[i]$ stamps when $T[i] \leq s$, while
we don't know how to generate $i$ otherwise. Table $U$ is initialized
to $1$, the goal being that at the end, if $T[i]<s$, $U[i]$ should contain the index $j$ of the greatest
value $a_j$ involved in any optimal decomposition of $i$.
We maintain a cursor starting at position $0$: the cursor at position $i$ means that we know the optimal number of stamps for all elements  $\leq i$.
We then gradually expand
our knowledge: if we can generate $i$ using $T[i]$ stamps, we can also
use one extra stamp (i.e.\ $T[i]+1$ stamps in total) to generate the
number $i+a_j$ for any $j \in \{U[i],\dots,k\}$ (note that we start the loop at $U[i]$, hence saving us from testing small stamps). If this combination is
better than the previously known combination to generate $i+a_j$, we
update its value, i.e.\ if $T[i+a_j]>T[i]+1$ we just need to compute
$T[i+a_j] \gets  T[i]+1$ and $U[i+a_j] \gets j$. It is easy to see that after this
operation, $T[i+1]$ now contains the optimal number of stamps to
generate $i+1$, where $T[i+1] > s$ means that it is impossible to
generate $i+1$ in less than $s$ steps. It is also possible to show (see \cref{prop:leoalgocorrect}) that $U[i]$ contains the highest possible stamp involved in any optimal decomposition of $i$.
This way, if $T[i+1] > s$, we
found the first element that cannot be generated with at most $s$
steps, hence the $s$-range is ${1,\dots,i}$. Otherwise we can just
restart this procedure from the new optimal construction of $i+1$
until it finishes (see~\cref{alg:leo}
in~\cref{appendix:proofsalgoforlargeenvelopes} for the details).

\begin{propositionE}[\ProofInAppendix{appendix:proofsalgoforlargeenvelopes}]\label{prop:leoalgocorrect}
  There exists an LPSP algorithm which
  requires $O(k s a_k)$ assignments or more precisely $O(kn_s(A_k)+sa_k)$.
\end{propositionE}
\begin{proofE}

  We exhibit the following \cref{alg:leo} and prove its correctness
  and complexity bounds.
  \begin{algorithm}[htbp]\caption{Incremental LPSP}\label{alg:leo}
    \begin{algorithmic}
      \REQUIRE Basis $A_k=\{a_1=1<a_2<\ldots<a_k\}$, stamps $s$.
      \ENSURE The $s$-range $n_s(A_k)$.
      \STATE $T \gets [0,s+1,s+1, \dots, s+1] \in \N^{(s+1) a_k+1 }$\COMMENT*{$s+1$ plays the role of $\infty$}
      \STATE $U \gets [1,1,1, \dots, 1] \in \N^{(s+1) a_k+1 }$
      \STATE $i \gets 0$
      \REPEAT
        \FOR{$j$ in $U[i]..k$}
          \IF{$T[i + a_j]>T[i]+1$}
            \STATE  $T[i + a_j] \gets  T[i] + 1$\COMMENT*{Changes when we find a better way to obtain $i + a_j$}
            \STATE $U[i+a_j] \gets j$
          \ENDIF
        \ENDFOR
      \STATE $i \gets i+1$
      \UNTIL{$T[i] > s$}
      \RETURN $i-1$
    \end{algorithmic}
  \end{algorithm}

  The correctness of  \cref{alg:leo} follows the intuition given in the introduction of \cref{subsec:incrementalLPSP}. More precisely, before showing the optimality of $T[i]$, we note that it is easy to show by a simple induction that the algorithm maintains at any time the following invariant: for any $i$, if $T[i] \leq s$, then it is possible to construct $i$ using $T[i]$ stamps where the decomposition involve at least once the stamp $a_{U[i]}$ (or no stamp in the special case of $0$), and where all other involved stamps are smaller or equal to $a_{U[i]}$: the initialization of $T$ is done so that all items but $T[0]$ are set to $s+1$ hence maintaining this invariant ($0$ is made of $0$ stamps), and the only instruction that modifies $T$ is $T[i + a_j] \gets T[i] + 1$ when $T[i] < T[i + a_j] \leq s + 1$, hence $T[i] < s$, so we can apply the induction hypothesis: by induction, if $T[i] < s$, it is possible to decompose $i$ into $T[i]$ stamps, hence by adding one more stamp, we can generate $T[i+a_j]$ with $T[i]+1$ stamps. Moreover, by induction, $T[i]$ can be decomposed using stamps smaller or equal to $U[i]$, and since $j > i$ and the basis is assumed to be sorted, the decomposition of $i + a_j$ uses the stamp $a_j = a_{U[i + a_j]}$ and this stamp is the higher stamp in the decomposition. Hence, the invariant is maintained.

  Next, we show by induction the optimality of this value, by showing that when starting the $t$-th iterations (indexing from $0$) of the main loop, $T[i]$ for $i \leq t$ contains $s+1$ if it is impossible to generate them with less than $s$ stamps, and otherwise it contains the optimal number of stamps that can be used to generate $i$, while $U[i]$ contains the index of the highest stamp involved in any possible optimal decomposition (and $1$ when $i=0$). This is trivial when $t = 0$ since $T[0] = 0$ as $0$ can be generated via $0$ stamps. Then, assume $t > 0$. We do now a case analysis:
  \begin{itemize}
  \item We first consider a first case where $t$ can be decomposed optimally (in term of number of stamps) into $t = \sum_i \lambda_i a_i$ for some $\lambda_i$'s with $\sum_i \lambda_i \leq s$, and such that this decomposition maximizes the highest stamp $a_l$ involved with a non-zero coefficient $\lambda_l > 0$ if multiple such optimal decomposition exist. Then $t - a_{l}$ can be generated via $(\sum_i \lambda_i) - 1$ stamps ($t - a_{l} = (\lambda_{l} - 1) a_{l} + \sum_{i \neq l} \lambda_i a_i$) and this decomposition is optimal (if $t - a_{l}$ could be generated from less stamps, then this directly gives a better decomposition in stamps for $t$ by simply adding the stamp $a_{l}$ to this better decomposition, which is absurd since the decomposition of $t$ was optimal). Hence by induction, we know that after the $t - a_{l}$-th loop, we have $T[t] = T[t-a_l] + 1$ (it cannot be smaller otherwise we could again find a better decomposition of $t$ which is absurd since we assumed it was optimal) corresponding to the optimal stamp decomposition of $t$. It is easy to see that this value cannot be changed in later iterations: it cannot be increased because of the test $T[i+a_j] > T[i] + 1$, and it cannot be decreased otherwise this would directly lead to a better decomposition of $t$. And it is also easy to see that $T[t]$ cannot have been set to its final value before the $t-a_l$-th iteration of the loop, because if it was set to its optimal value during iteration $t-x$ with $x > a_l$, it means that we could extract from it an optimal decomposition involving the stamp $x$, which is absurd since $x > a_l$ and we we picked the optimal decomposition that was maximizing the highest possible stamp. Hence, it means that $U[i]$ is set during the $t-a_l$-th iteration of the loop to $a_l$, the highest possible stamp involved in the decomposition, concluding our induction for this case.
  \item The second case is when $t$ cannot be decomposed into $s$ stamps or less. But thanks to the invariant shown in the first part of the proof, we already know that $T[t] = s + 1$, concluding this case analysis and this induction proof.
  \end{itemize}

  The complexity bound is trivial to prove: the main loop repeats until $i$ equals $n+1$, and $i$ is incremented by $1$ at each step, hence it repeats $O(n)$ times, while doing $k$ assignments during the \emph{for} loop, i.e.\ overall $O(kn)$ assignments. But trivially $n \leq s a_k$ (the highest number that can be generated places the maximum number of stamps with the maximum value). Moreover, the creation of $T$ requires $O(s a_k)$ assignments, hence the overall number of assignments is $O(k s a_k)$.
\end{proofE}


\subsection{Sliding window}

We describe in \cref{alg:sliding} an improvement of \cref{alg:leo} that is more space and time efficient, the size of the memory $O(a_k)$ scaling only with the maximum stamp value $a_k$, and we also avoid to initialize $T$ to a too large value $s a_k$ approximating $n$, hence replacing the dependency on $s a_k$ into a dependency on $n$. This algorithm may also perform better due to better cache optimizations since the table is now smaller and may even entirely fit in the cache.

\begin{algorithm}[ht]\caption{Sliding window incremental LPSP}\label{alg:sliding}
  \begin{algorithmic}
    \REQUIRE Basis $A_k=\{a_1=1<a_2<\ldots<a_k \}$, stamps $s$.
    \ENSURE The $s$-range $n_s(A_k)$.
    \STATE $l \gets \lceil \log (a_k+1) \rceil$ \COMMENT*{The size of the array is a power of $2$ to have efficient modulo}
    \STATE $w \gets 2^l - 1$ \COMMENT*{We compute $x \bmod l$ via the bitwise ``and'' $x \bitwiseand w$ for efficiency reasons}
    \STATE $T \gets [0,s+1,s+1,\dots,s+1] \in \N^{2^l}$
 \STATE $U \gets [1,1,1, \dots, 1] \in \N^{2^l }$
    \STATE $i \gets 0$
    \REPEAT
     \FOR{$j$ in $U[i]..k$}
\IF{$T[(i + a_j ) \bitwiseand w]>T[i \bitwiseand w]+1$}
     \STATE  $T[(i + a_j) \bitwiseand w] \gets  T[i \bitwiseand w] + 1$
\STATE $U[(i + a_j) \bitwiseand w] \gets j$
\ENDIF
     \ENDFOR
     \STATE $T[i \bitwiseand w] \gets s+1$
     \STATE $i \gets i+1$
    \UNTIL{$T[i \bitwiseand w ] > s$}
    \RETURN $i-1$
  \end{algorithmic}
\end{algorithm}

The improvement of~\cref{alg:sliding} over~\cref{alg:leo} comes from the main observation that in \cref{alg:leo} all operations on the array lies in a window of size $a_k+1$. Hence, we can only keep this window and discard the rest of the array by considering all indices modulo $a_k + 1$. To improve the efficiency of the algorithm, we additionally consider a larger window whose size is a power of $2$. This way, modulo operations can be implemented via a simple bitwise ``\emph{and}''.
We note that this algorithm is also slightly more efficient it term of
time due to the fact that it does not need to create a large table of
size $s a_k$ (a rough estimation of $n$). This way, the complexity of
this algorithm scales with $n$ instead of $s a_k$, which may be
interesting when $n \ll s a_k$.

\begin{propositionE}[\ProofInAppendix{appendix:proofsalgoforlargeenvelopes}]
  \Cref{alg:sliding} is correct and requires $O(k s a_k)$ assignments or more precisely $O(k n_s(k) + a_k)$.
\end{propositionE}
\begin{proofE}
  The correction of this algorithm is a direct consequence of the correction of \cref{prop:leoalgocorrect}, combined with the fact that in \cref{prop:leoalgocorrect} after $i$ iterations the loop only modified elements in the table whose index is smaller than $i + a_k$, and will never access again elements whose index is smaller than $i$. Hence, we can write the new values on top of the items that will never be read again (and initialize them as we go), which is exactly the point of doing a modulo with a modulus $2^l$ larger than $a_k + 1$, implemented via the bitwise \emph{and} operation with $w = 2^l - 1$ for efficiency reasons.

  The complexity bound is analog to \cref{prop:leoalgocorrect}, except that the creation of the original table does $O(a_k)$ assignments instead of $O(s a_k)$, while the loop still stops after $n$ iterations doing $k$ assignments at each iteration.
\end{proofE}

\subsection{Comparison to the state of the art}

Table \ref{tab:comp} compares the asymptotical complexity of our two algorithms
(\cref{alg:sliding,alg:leo}) with \cite{Mossige:1981:hrange}.

\begin{table}[ht]\centering
\caption{Comparison of asymptotic complexity}\label{tab:comp}
\begin{tabular}{cccc}
\toprule
 & \cite{Mossige:1981:hrange} &\cref{alg:leo} &  \cref{alg:sliding} \\
\midrule
Time &  $O(ks^2a_k)$ &$O(ks a_k)$& $O(ks a_k)$\\
Memory & $O(sa_k)$& $O(sa_k)$ & $O(a_k)$
\\
\bottomrule
\end{tabular}
\end{table}

 We also compare those three algorithms based on
numerical simulations to also assert the practicality of our
approach.
\cite{Mossige:1981:hrange} remains faster only for very small values
of $s$ (smaller than $5$ in our simulations)
as shown in \cref{fig:cover_s05_s10} for $s=5$ and $s=10$,
or \cref{fig:cover_k15,fig:cover_k35} for various $s$.
Note that a the early termination condition of Selmer's Lemma
(see~\cite{Selmer:1980:Postage,Selmer:1983:Postage,Selmer:1985:Postage})
can be added to all the presented range algorithms :
when the condition of the Lemma is
met during the computation, the range is then known~\cite[(10)]{Mossige:1999:Postage}.
For instance, \cref{fig:cover_k15} suggests a speed-up growing linearly
with the number of stamps, that reaches a threshold when
Selmer's Lemma is used to early terminate the range computations.
%
%

\textEnd{
The range algorithms are implemented in the library~\GStamps, and can
be respectively used via \texttt{./bin/srange} for~\cref{alg:mossige},
\texttt{./bin/reach} for~\cref{alg:leo} and \texttt{./bin/krange}
for~\cref{alg:sliding}.

  \begin{figure}[htbp]\centering
    \input{data_cover_s05_s0}

    \input{data_cover_s10_s0}
    \caption{LPSP algorithms,~$s=5$ and $s=10$}
    \label{fig:cover_s05_s10}
  \end{figure}

  \begin{figure}[htbp]\centering
    \input{data_cover_k15_s0}
    \input{data_cover_k15_s1}
    \caption{LPSP algorithms,~$k=15$, early terminated by Selmer's
      lemma~\cite{Selmer:1980:Postage} (effective for $s\geq{40}$).}\label{fig:cover_k15}
  \end{figure}


  \begin{figure}[htbp]\centering
    \input{data_cover_k35_s0}
    \caption{LPSP algorithms,~$k=35$}\label{fig:cover_k35}
  \end{figure}
}

\section{Applications to privacy-preserving cryptographic protocols}\label{sec:applications}
\pratendSetLocal{category=applications}

Private set operation protocols are cryptographic protocols in which two (or more) parties aim to perform some set operations while preserving privacy. For instance in a private set union protocol, a sender holds a set $X$ and a receiver a set $Y$, and the goal is for the receiver to learn the union $X\cup Y$ while it should not learn the intersection $X\cap Y$. Many such protocols, for instance for the union~\cite{DBLP:conf/acns/DumasGGMR25} or the intersection~\cite{DBLP:conf/ccs/ChenLR17,cong_labeled_2021}, rely on polynomial evaluation : The client owns some value $x$ in $\mathbb F$ and the server a degree-$d$ polynomial $P\in\mathbb F[X]$, and the goal is for the client to learn $P(x)$ while neither does the server learns $x$ nor the receiver $P$.

An easy solution for this problem is for the client to encrypt $x$ using a fully homomorphic encryption scheme (FHE)\footnote{Details on fully homomorphic encryption schemes are presented in \cref{app:fhe}.} and send it to the server. The later evaluates its polynomial~$P$ on the encryption of $x$ using the homomorphic properties of the encryption scheme. It obtains an encryption of $P(x)$ that the client can then decrypt to obtain $P(x)$.
A limitation of this solution is the cost of homomorphic
operations. In particular, the efficiency of homomorphic operations is
directly related to the \emph{multiplicative depth} of the function to
evaluate. The evaluation of a degree-$d$ polynomial requires a
multiplicative depth $\log(d)$, which becomes soon the bottleneck in practice.

To reduce the multiplicative depth of polynomial evaluation, the idea
is for the client to send not only its (encrypted) value $x$, but also
some of its powers $x^{a_i}$, $i\in[k]$.
Viewing $A_k = \{a_1 < \dotsb < a_k\}$ as a stamp basis, $n_{2^\ell}(A_k) \ge d$ implies that the server can evaluate any degree-$d$ polynomial $P$ with a multiplicative depth $\ell+1$. Indeed, this means by definition that any power $x^d$ is a product of at most $2^\ell$ powers $x^{a_i}$, which can be computed in depth $\ell$ using a binary tree. The extra depth allows for the final inner product of the coefficients of $P$ by the powers of $x$. As an example, if the client sends $x$, $x^5$ and $x^8$, the server can evaluate any degree-$26$ polynomial in depth $3$ since $n_4(\{1,5,8\}) = 26$. Indeed, $x^{26} = (x^5\cdot x^5)\cdot(x^8\cdot x^8)$.

This idea leads to a trade-off between communication (since several
powers are sent) and multiplicative depth, hence efficiency in
practice. The trade-off is actually even better than it may seem at first sight. Indeed, reducing the multiplicative depth also reduces the size of each ciphertext. Instead of sending one large ciphertext, the sender sends $k$ smaller ciphertexts. The communication still increases, but by a factor less than $k$.
As far as we know, the explicit use of a stamp basis was first suggested by~\cite{cong_labeled_2021} to build an efficient private set intersection protocol ; although no study on the practical efficiency of this approach was provided. We provide such an analysis that confirms the relevance of this approach.
We use our library \GStamps to find small bases allowing to evaluate a
polynomial with a fixed multiplicative depth $L$
for various degrees $d$ as shown in~\cref{tab:stamps}. We denote those constructions as \textbf{base-L} for the use of a basis allowing to do the protocol in depth $L$.
As an example, if we want to perform the protocol with $d=2^{20}$ and under a maximum depth $L=7$, we use our algorithm \texttt{./bin/search} with inputs $s=2^{L-1}$ and $N=d$. It returns the stamp basis $(1,52,705,13100,99644)$ which has a $s$-range of $1782370\geq2^{20}$.
\begin{table}[htbp]
\centering
\caption{$k=\#$denominations for basis obtained with \texttt{./bin/search}}\label{tab:stamps}
\begin{tabular}{|l|r|r|r|r|r|r|r|}
\hline
Constr.&s&k for $d=2^{10}$&$d=2^{12}$&$d=2^{14}$&$d=2^{16}$&$d=2^{18}$&$d=2^{20}$\\
\hline
\textbf{base-5}&16&4&5&6&9&10&12\\
\hline
\textbf{base-6}&32&3&4&5&5&7&8\\
\hline
\textbf{base-7}&64&2&3&4&4&5&\bf 5\\
\hline
\end{tabular}
\end{table}

We compare several protocols for polynomial evaluation using different communication-depth trade-offs in~\cref{tab:Esti}. The \textbf{naive} construction consists in sending only the value $x$, for a multiplicative depth $\log(d)$. A second construction, following~\cite{DBLP:conf/ccs/ChenLR17,DBLP:conf/ccs/Tu0LZ23}, reduces the depth to $O(\log\log d)$ by sending sends encryptions of $x^{2^i}$, $i\in[\log d]$. In fact, this can be viewed as a geometric progression stamp basis with common ratio $2$~(\cref{rem:geom}). We denote this as the \textbf{geom.} construction. 

We compare the communication volumes and the expected running times to homomorphically compute all the powers for each construction for various polynomial degrees.  The communication volumes presented in~\cref{tab:Esti} take into account the number of ciphertexts as well as their sizes. 
As it happens, not all homomorphic multiplication have the same running time. \emph{Deeper} multiplications are faster. 
The timing estimations in~\cref{tab:Esti} take this aspect into account. 
Although it is easy to exactly count the number of multiplications per depth for the \textbf{naive} and \textbf{geom.} constructions, we rely on our program \texttt{./bin/depthrange} from our library for the \textbf{base-L} constructions.

\begin{table}[htbp]
\centering
\caption{Communication volumes (MB) and runtime (s) estimation for various polynomial degrees}\label{tab:Esti}
\begin{tabular}{|l|l|l|l|l|l|l|l|l|l|l|l|l|}
\hline
\multirow{2}{*}{Constr.}&\multicolumn{2}{c|}{$d=2^{10}$}&\multicolumn{2}{c|}{$d=2^{12}$}&\multicolumn{2}{c|}{$d=2^{14}$}&\multicolumn{2}{c|}{$d=2^{16}$}&\multicolumn{2}{c|}{$d=2^{18}$}&\multicolumn{2}{c|}{$d=2^{20}$}\\
\cline{2-13}
&MB&s&MB&s&MB&s&MB&s&MB&s&MB&s\\
\hline
\textbf{naive}&\bf 3.7&36&\bf 4.9&173&\bf 5.0&763&\bf 5.5&3590&\bf 6.1&16032&\bf 6.8&69830\\
\hline
\textbf{geom.}&18.3&26&21.9&100&25.5&382&33.9&1745&38.1&6674&42.3&25676\\
\hline
\textbf{base-5}&7.5&21&9.3&\bf 81&11.1&\bf 311&16.5&1280&18.3&\bf 4898&21.9&19375\\
\hline
\textbf{base-6}&6.6&24&8.7&100&10.8&404&10.8&\bf 1190&15.0&5466&17.0&20683\\
\hline
\textbf{base-7}&5.0&\bf 20& 7.3&98& 9.7&504& 9.7&1315&12.0&5940&12.0&\bf 18916\\
\hline
\end{tabular}
\end{table}

We observe in~\cref{tab:Esti} the relevance of the trade-off for the performance. The use of stamps in geometric progression is worse both in communication, as expected, but also in runtime compared to the use of better bases. We can see that a modest increase in communication, with a factor less than two, can result in a huge runtime gain, almost a $3.7\times$ speedup. As an example, consider the degree $d=2^{20}$. To allow a depth $21$ as required in the naive construction, one has to send a ciphertext of size $6.55$ MB. The answer of the receiver has size $0.25$ MB, for a total of $6.8$ MB. Using instead the construction \textbf{base-7}, the basis has $5$ denominations. Since the required depth is only $7$, each ciphertext has size $2.35$ MB, for a total of $12$ MB. 
To conclude, the use of good stamp bases allows to choose more finely the different parameters of the FHE scheme, and it results in a practical efficiency boost. 
A good trade off between runtime and communication can, suprisingly, be more efficient both in runtime and communication.

\clearpage
\bibliographystyle{plainurl}
\bibliography{bibstamps.bib}

\appendix
\renewcommand{\ProofInAppendix}[1]{}
\section{Appendix}

\subsection{Fully homomorphic encryption (FHE)}\label{app:fhe}

A fully homomorphic encryption scheme is a semantically secure
public-key encryption scheme that allow to perform arithmetic
operation on ciphertexts resulting, after decryption, to operations on
the underlying plaintexts. Namely, a FHE is built with a \textit{key
  generation} algorithm, generating private and public keys, an
\textit{encryption} and a \textit{decryption} algorithm ; the first
takes a plaintext and the public key to output a ciphertext, while the
second inputs a ciphertext and the secret key to output a
plaintext. Basically, the scheme allows to perform \textit{homomorphic
  addition}, \textit{plaintext-ciphertext product} and
\textit{homomorphic product} on ciphertexts, whose result after
decryption to the corresponding operations on plaintexts.

Most of the existing FHE schemes
(TFHE~\cite{DBLP:journals/joc/ChillottiGGI20},
FHEW\cite{DBLP:conf/eurocrypt/DucasM15},
CKKS~\cite{DBLP:conf/asiacrypt/CheonKKS17},
BFV~\cite{DBLP:conf/crypto/BrakerskiV11},
BGV~\cite{DBLP:journals/toct/BrakerskiGV14}) security rely on the
difficulty of the LWE game ; it basically means that a ciphertext is
masked with a noise. Each homomorphic operation increase the size of
this noise and the decryption remains correct as long as this size
does not exceed a threshold. For that reason there exist some
procedures to control the noise growth. The most expensive operation
in term of noise expansion is the homomorphic product ; for that
reason, the focus is put essentially on reducing the amount of
homomorphic products, in particular in sequence. In the ring FHE
schemes (BGV, BFV, CKKS), each homomorphic product is followed by a
\textit{modulus switching} procedure ; the idea is to reduce the size
of the ciphertext space in order to reduce the size of the noise. This
can be done a finite number of time, let say $L$ times, as long as the
ciphertext remains secure. In theory, this $L$ can be as big as
needed, as long as the fresh ciphertext space is big enough. In
practice, setting up a context allowing $10$ modulus switchings, for
example, already implies computational cost overheads, a huge
communication volume to exchange ciphertexts, and possibly can
threaten the security of the scheme. We remark that the amount of
modulus switching corresponds to the multiplicative depth of the
circuit. From now on, we consider that $L$ is the \textit{maximum
  multiplicative depth} allowed by the scheme.

\begin{figure}[htbp]\centering
\input{ctxt_size}
\caption{Experimental results for various maximum depth using openFHE}\label{fig:perdepth}
\end{figure}

\Cref{fig:perdepth} shows some experimental results obtained with an
Intel(R) Xeon(R) Gold 6330 CPU @ 2.00GHz on a single thread using the
BFV scheme from the openFHE
library\footnote{\url{https://github.com/openfheorg/openfhe-development}}.
We generated contexts for various $L$, while keeping the same
plaintext space, which is $\mathbb{F}_p$ with $p$ on $48$ bits, and
with a security level beyond 128 bits. \cref{fig:perdepth} shows that
both the size of a fresh ciphertext and the runtime of one
homomomorphic product on fresh ciphertexts\footnote{The timings are
  obtained as a mean of 100 products.} grow with the maximum depth
allowed by the scheme.

\begin{remark}\label{rem:timingsperslots}
    In practice, one ring FHE ciphertext is big enough to encrypt many
    plaintexts through batching. We denote that as the number of
    \textit{slots} of the ciphertext. Performing a homomorphic
    operation on a ciphertext is forced to act simultaneously on all
    the slots. The different contexts may have various number of slots
    depending on the $L$ chosen, due to the fact that we impose a
    certain security level. Namely, the contexts with $L\in[2,5]$ have
    $2^{14}$ slots, for $L\in[6,11]$ there are $2^{15}$ slots, and for
    bigger $L$ (up to 21 in our experiments), we have $2^{16}$
    slots. For consistency of the experiments, all the experimental
    results given in \cref{sec:applications} are runtimes and volumes
    \textit{per $2^{14}$ slots}. For example, the sizes and runtimes
    for schemes with $L\geq 12$ are, in real life, $4$ times bigger
    that the one in~\cref{fig:perdepth}, but we can technically do $4$
    times the job using SIMD. An important remark is that if not all
    the slots are used, then the results for $L\in[6,11]$ and, even
    more, for $L\geq 12$ are much worse.
\end{remark}

Once the \textit{modulus switching} procedure is no longer doable, or
after each homomorphic products for the FHE schemes on the Torus
(FHEW, TFHE), the \textit{bootstrapping} procedure can be
done~\cite{DBLP:conf/eurocrypt/DucasM15}; it basically consists in a
homomorphic decryption which resets the noise to (almost) the fresh
level. This procedure is extremely expansive in the ring schemes and
requires a specific choice of parameters ; in practice, it is often
tried to avoid this. For those reasons, we can understand how reducing
or atleast controlling the multiplicative depth can matter in
practice, especially for the ring schemes that we are focusing on.



\subsection{Proofs of \cref{sec:polytimeapprox}} \label{sec:proofsapprox}


We first provide the proof of the $s$-range of Alter and Barnett's
construction~\cite{alter_remarks_1977}. We generalize their construction as
follows.

Let $q_1$, \dots, $q_s \ge 1$. We define a stamp basis as the union of $s$
blocks of size $q_1$, \dots, $q_s$. Each block is an arithmetic progression.
The first block is $B_1 = \{1,\dotsc,q_1\}$. Assume that the first $i$ blocks
have been built. Then, the value common difference of the $(i+1)$st block is the
next value in the arithmetic progression of the $i$th block. And the difference
between the largest value of the $i$th block and the smallest of the $(i+1)$st
block is the common difference of the $(i+1)$st block.

That is, let $B_i = \{u_i + t v_i:0\le t < q_i\}$ for $i\in[s]$, where $u_1 =
v_1 = 1$ and $v_{i+1} = u_i+q_iv_i$ and $u_{i+1} = u_i+(q_i-1)v_i+v_{i+1}$.
Define the stamp basis associated to $q_1$, \dots, $q_s$ as $G(q_1,\dotsc,q_s) =
B_1\cup\cdots\cup B_s$.

Alter and Barnett's constructions is the case $q_1 = \dotsb = q_{s-1} = q$ and
$q_s = q+r$ where $k = qs+r$ is the number of denomination. Our improvement is
the case $q_1 = \dotsb = q_{s-r} = q$ and $q_{s-r+1} = \dotsb = q_s = q+1$.

\begin{proposition}
    Let $q_1$, \dots, $q_s \ge 2$ and $G = G(q_1,\dotsc,q_s)$. Then
    $n_s(G) = u_s + q_sv_s-1$.
\end{proposition}
\begin{proof}
    Let $v_{s+1} = u_s+q_sv_s$, $u_{s+1} = u_s+(q_s-1)v_s+v_{s+1}$ and for $1\le
    i\le s$, $G_i = G(q_1,\dotsc,q_i)$. We start by proving, by induction on
    $i$, that $n_i(G_i) \ge v_{i+1}-1$ and $n_{i+1}(G_i)\ge u_{i+1}-1$.  The
    base case $i = 1$ is clear since $G_1 = \{1,\dotsc,q_1\}$, $v_2 = 1+q_1$ and
    $u_2 = 2q_1+1$. Assume now that $n_{i-1}(G_{i-1})\ge v_i-1$ and
    $n_i(G_{i-1})\ge u_i-1$ for some $i\ge 2$. Any integer $m < u_i$ can be
    written with $i$ stamps from $G_{i-1}\subset G_i$. Consider an integer
    $m\in\{u_i,\dotsc,v_{i+1}-1\}$. Since $v_{i+1} = u_i+q_iv_i$, if we take the
    largest possible stamp $u_i+tv_i\le m$ of $G_i$, $m-(u_i+tv_i) < v_i$.
    Therefore, $m-(u_i+tv_i)$ can be written using $(i-1)$ stamps from $G_{i-1}$
    and $m$ can be written with $i$ stamps from $G_i$. Thus $n_i(G_i) \ge
    v_{i+1}-1$. Now take any $m\in\{v_{i+1},\dotsc,u_{i+1}-1\}$. Since $u_{i+1}
    = u_i+(q_i-1)v_i+v_{i+1}$, the same argument as before proves that using the
    largest possible stamp $u_i+tv_i\le m$, $m-(u_i+tv_i) < v_{i+1}$.
    Therefore, $m-(u_i+tv_i)$ can be written with $i$ stamps from $G_i$, and $m$
    with $(i+1)$ stamps from $G_i$. That is, $n_{i+1}(G_i) \ge u_{i+1}-1$.

    Let us now prove that $n_s(G_s) < v_{s+1} = u_s+q_sv_s$, that is $v_{s+1}$
    cannot be written as a sum of at most $s$ stamps from $G$.

    A stamp in block $i$ has value $u_i + tv_i$ for some $t$. Grouping all
    stamps in a same block, we can write the sum of their values as $a_iu_i +
    b_iv_i$. A sum corresponding to at most $s$ stamps from $G$ is $\sum_{i=1}^s
    a_iu_i + b_iv_i$ with $\sum_i a_i\le s$ and $b_i \le (q_i-1)a_i$ for all
    $i$. We identify the sum with the sequence $((a_i,b_i))_i$.

    We will prove that given such a sequence, we can produce a new sequence
    $((a'_i,b'_i))_i$, with the same sum, with the extra property that for every
    $i >0$, $a'_i\le 1$ and $a'_i+b'_i\le q_i$.  To this end, we use identities:
    \begin{enumerate}
        \item $u_i+q_iv_i = v_{i+1}$;
        \item $2v_i = u_i+v_{i-1}$;
        \item $u_i = u_{i-1}+(q_{i-1}-1)v_{i-1}+v_i$.
    \end{enumerate}
    The first and third ones are the definition of the sequences. The second one
    combines the two other ones: $u_i = u_{i-1}+(q_{i-1}-1)v_{i-1}+v_i$, and
    $u_{i-1}+q_{i-1}v_{i-1} = v_i$, whence $u_i = (v_i-v_{i-1})+v_i$.

    These relations allows us to define three rewriting rules on the sequence
    $((a_i,b_i))_i$ that does not change the value $\sum_i a_iu_i+b_iv_i$:
    \begin{enumerate}
        \item if $a_i>0$ and $b_i\ge q_i$, replace $(a_i,b_i),(a_{i+1},b_{i+1})$
            by $(a_i-1,b_i-q_i),(a_{i+1},b_{i+1}+1)$ (using the first identity);
        \item if $a_i = 0$ and $b_i>q_i$, replace $(a_{i-1},b_{i-1}),(a_i,b_i)$
            by $(a_{i-1},b_{i-1}+1),(a_i+1,b_i-2)$ (using the second identity);
        \item if $a_i > 1$ and $b_i<q_i$, replace $(a_{i-1},b_{i-1}),(a_i,b_i)$
            by $(a_{i-1}+1,b_{i-1}+(q_i-1)),(a_i-1,b_i+1)$ (using the third
            identity).
    \end{enumerate}
    We will apply these rules with $i > 0$. Note that if $i = s$, we may have to
    consider $(a_{s+1},b_{s+1})$. This is possible by extending the sequences
    $(u_i)_i$ and $(v_i)_i$ with $v_{s+1} = u_s+q_sv_s$ and $u_{s+1} =
    u_s+(q_s-1)v_s+v_{s+1}$. We can choose any very large value for $q_{s+1}$ to
    never apply the first or second rule with $i = s+1$.

    Starting from a sum $\sum_i a_iu_i+b_iv_i$, our goal is to use the rewriting
    rules to obtain a sum where for all $i>0$, $a_i \le 1$ and $a_i+b_i \le
    q_i$. The starting point is a sum where $\sum_i a_i \le s$, and $b_i \le
    (q_i-1)a_i$ for each $i$. We call a sequence $((a_i,b_i))_{i\ge 0}$
    \emph{valid} if $a_i+b_i\le\max(1,a_i)q_i$ and $\sum_i a_i \le s$. In
    particular, the starting sequence is valid.

    We choose the largest index $j$ that violates the conditions, that is
    $a_j > 1$ or $a_j+b_j > q_j$. If $j > 1$, we apply the first of the three rules that
    applies to $(a_j,b_j)$, and recurse. We stop when $a_i\le 1$ and $a_i+b_i\le q_i$
    for all $i>0$. We need to prove that this process terminates.

    First assume that the first rule is applied to a valid sequence. Since this
    replaces $(a_{j+1},b_{j+1})$ by $(a_{j+1},b_{j+1}+1)$, the sequence
    becomes invalid if $a_{j+1}+b_{j+1} = q_{j+1}$. (Recall that $a_{j+1} \le 1$ by
    maximality of $j$.) If $a_{j+1} = 1$, the first rule is again applied to
    replace $(1,q_{j+1})$ by $(0,0)$ and $(a_{j+2},b_{j+2})$ by $(a_{j+2},b_{j+2}+1)$.
    The situation may repeat a few times, but after a few iterations the
    rewriting rule will encounter a pair $(a_{j+t},b_{j+t})$ where either
    $a_{j+t}+b_{j+t} < q_{j+t}$ or $a_{j+t} = 0$. In the first case, the sequence is valid
    again. In the second case (with $b_{j+1} = q_{j+t}+1$), the second rule is applied and
    the sequence becomes valid again too. In both cases, after these few steps,
    the sequence is valid and either $j$ has decreased or $a_j$ has decreased.
    Note also that to apply the first rule to a valid sequence, $a_j$ has to be
    at least $2$ (since $q_j\le b_j \le (q_j-1)a_j$). After these steps, $a_j$
    remains non-zero.
    The same phenomenon occurs when the third rule is applied to a valid sequence.
    It replaces $(a_{j-1},b_{j-1}),(a_j,b_j)$ by
    $(a_{j-1}+1,b_{j-1}+(q_{j-1}-1)),(a_j-1,b_j+1)$, with the condition that $a_j>1$
    and $b_j < q_j$. The sequence becomes invalid if $a_j = 2$ and $b_j=q_j-1$.
    But then the pair becomes $(1,q_j)$ and the next rule to be applied is the first
    rule. As before, after a few steps, the sequence becomes valid again and
    either $j$ or $a_j$ has decreased, and $a_{j-1}$ is non-zero .
    The second rule only applies to invalid sequences. By grouping the rewriting
    steps into \emph{meta-steps}, we get a rewriting process where the sequence
    stays valid and the pair $(j,a_j)$ decreases (for the lexicographic order)
    at each meta-step. The process thus stops with a valid sequence.

    Assume therefore that $v_{s+1}= u_s+q_sv_s$ admits a representation with $s$ stamps. Then
    $v_{s+1} = \sum_i a_iu_i+b_iv_i$ where $a_i\in\{0,1\}$ and
    $0\le b_i < q_i$ for $i >0$,
    and $\sum_i a_i \le s$. 
    After each meta-step, $a_i\neq 0$ for at least one index $i \le j$. In particular, 
    $\sum_{i=1}^s a_iu_i+b_iv_i > 0$ at the end of the process.
    Therefore, the pairs $(a_i,b_i)$, $i > s$, must be zero. Assume otherwise
    that $(a_j,b_j)\neq(0,0)$ for some $j > s$. The total sum is then at least
    $(a_ju_j+v_jb_j) + \sum_{i=1}^s a_iu_i+b_iv_i > a_ju_j+b_jv_j$. Since
    $u_j$, $v_j\ge v_{s+1}$, the sum is $> v_{s+1}$, a contradiction.
    Thus $v_{s+1} = \sum_{i=1}^s a_iu_i+b_iv_i$, with
    $\sum_i a_i \le s$. By validity of the sequence $a_iu_i+b_iv_i \le
    u_i+(q_i-1)v_i$ for all $i$. 
    The largest possible sum is then
    $\sum_{i=1}^s u_i+(q_i-1)v_i = \sum_{i=1}^s v_{i+1}-v_i = v_{s+1}-1$.
    Therefore, $v_{s+1}$ cannot be reached and $n_s(G_k) \leq v_{s+1}-1$.
\end{proof}

\textEnd[category=polytimeapprox]{} 
\printProofs[polytimeapprox]

\subsection{Proofs and figures of \cref{sec:algoforlargeenvelopes}}\label{appendix:proofsalgoforlargeenvelopes}

\textEnd[category=algoforlargeenvelopes]{}
\printProofs[algoforlargeenvelopes]



\end{document}